\documentclass[11pt,a 4paper]{article}
\usepackage[T1]{fontenc}%
\usepackage[utf8]{inputenc}%
\usepackage{lmodern}%
\usepackage{textcomp}%
\usepackage{fullpage}%
\usepackage{color}%
\usepackage{xspace}%
\usepackage{ragged2e}%
\usepackage{graphicx}%
\usepackage{hyperref}
\usepackage{amsmath}
\usepackage[]{listings}
\usepackage{todonotes}
\usepackage{epstopdf} 

\usepackage[english]{babel}
\usepackage{caption}
\usepackage[utf8]{inputenc}
\usepackage[T1]{fontenc}
\usepackage{lmodern}
\usepackage{amssymb,amsmath,amsthm,amsfonts}
\usepackage[margin=2.4cm]{geometry}
\usepackage{enumitem}
\usepackage{color}
\usepackage{cite}
\usepackage{hyperref}
\usepackage{indentfirst}
\usepackage{todonotes}

\marginparwidth=\dimexpr \marginparwidth + 0.6cm\relax

\newlist{shortlist}{itemize}{1}
\setlist[shortlist]{label=\textbullet,nosep}

\newtheorem{thm}{Theorem}

\newtheorem{cor}[thm]{Corollary}
\newtheorem{conj}[thm]{Conjecture}

\theoremstyle{definition}

\newtheorem*{conj*}{Conjecture}

\newcommand{\etal}{{\emph{et al.}}}

\newcommand{\Path}[1]{\ensuremath{{\sf P}_{#1}}}          
\newcommand{\C}[1]{\ensuremath{{\sf C}_{#1}}}             
\newcommand{\K}[1]{\ensuremath{{\sf K}_{#1}}}             
\newcommand{\M}[1]{\ensuremath{{\sf M}_{#1}}}             
\newcommand{\extG}[2]{\ensuremath{{\sf E}_{#1,#2}}}    
\newcommand{\extg}[3]{\ensuremath{{\sf E}_{#1,#2,#3}}}    
\newcommand{\St}[2]{\ensuremath{{\sf S}_{#1,#2}}}    
\newcommand{\extH}[2]{\ensuremath{{\sf H}_{#1,#2}}}    

\newcommand{\dist}{\ensuremath{dist}}              

\makeatletter

\title{Minimum Eccentric Connectivity Index for Graphs\\with Fixed
	Order and Fixed Number of Pending Vertices}
\author{
	\large
	\textsc{
		Gauvain Devillez\textsuperscript{1},
		Alain Hertz\textsuperscript{2}%
		\footnote{Corresponding author : email  alain.hertz@gerad.ca;  tel. +1-514 340 6053.},
		}\\
	\textsc{
		Hadrien M\'elot\textsuperscript{1},
		Pierre Hauweele\textsuperscript{1}}\\[3mm]
	\normalsize \textsuperscript{1} Computer Science Department - Algorithms Lab\\
	\normalsize University of Mons, Mons, Belgium\\[3mm]
	\normalsize \textsuperscript{2} Department of Mathematics and Industrial
	Engineering\\
	\normalsize Polytechnique Montréal - Gerad, Montréal, Canada
}

\newcommand*{\eci}{\ensuremath{\xi^c}\xspace}
\newcommand*{\ecc}[1]{\ensuremath{e_{G}(#1)}\xspace}
\newcommand*{\eccp}[1]{\ensuremath{e_{G'}(#1)}\xspace}

\newcommand*{\degr}[1]{\ensuremath{d_{G}(#1)}\xspace}
\newcommand*{\degrp}[1]{\ensuremath{d_{G'}(#1)}\xspace}

\newcommand*{\sizs}[1]{\ensuremath{|#1|}\xspace}

\begin{document}
	\normalsize%
	\maketitle

\vspace{-2.5ex}
	\begin{abstract}
		The eccentric connectivity
		index of a connected graph $G$ is the sum over all vertices $v$ of the product $\degr{v}\ecc{v}$, where $\degr{v}$ is the degree of $v$ in $G$ and $\ecc{v}$ is the maximum distance
		between $v$ and any other vertex of $G$. This index is helpful for the prediction of biological activities 
		of diverse nature, a molecule being modeled as a graph where 
		atoms are represented by vertices and 
		chemical bonds by edges.
		We characterize those graphs which
		have the smallest eccentric connectivity index among all connected graphs of a given order $n$. Also, given two integers $n$ and $p$ with $p\leq n-1$, we characterize those graphs which have the smallest
		eccentric connectivity index among all connected graphs of order $n$ with $p$ pending vertices.
	\end{abstract}
	\section{Introduction}
	
	A chemical graph is a representation of the structural formula of a chemical compound in terms of graph theory where atoms are represented by vertices and chemical bonds by edges. Arthur Cayley \cite{Cayley1874} was probably the first to publish results that consider chemical graphs. 
	In an attempt to analyze the chemical properties of alkanes, Wiener \cite{Wie47} has introcuced the \emph{path number index}, nowadays called  \emph{Wiener index}, which is 
	defined as the sum of the lengths of the shortest paths between all pairs of vertices. Mathematical properties and chemical applications of this distance-based index have been widely researched.

Numerous  other topological  indices are used for  quantitative  structure-property  relationship (QSPR)  and quantitative  structure-activity relationship (QSAR) studies that help to describe and understand the structure of molecules \cite{Tod00, Kar00}, among which  the \emph{eccentric connectivity index} which can be defined as follows.
		Let $G=(V,E)$ be a simple connected undirected graph. The \emph{distance} $\dist_G(u,v)$
	between two vertices $u$ and $v$ in $G$ is the number of edges of a shortest path in $G$ connecting $u$ and $v$. 
	The \emph{eccentricity}
	$\ecc{v}$ of a vertex $v$ is the maximum distance between $v$ and
	any other vertex, that is $\max \{ \dist_G(v, w) ~|~ w \in V \}$. 
	The \emph{eccentric connectivity index} $\eci(G)$ of $G$ is defined by
	\[
	\eci(G) = \sum_{v \in V} \degr{v} \ecc{v}.
	\]
	This index was introduced by Sharma \etal~in~\cite{Sharma97} and successfully used for mathematical models of biological activities of diverse nature \cite{Duj08, Gup02, Kumar2004, Sar01, Ilic11}.
	Recently, Hauweele \etal \cite{Hau18} have characterized those graphs which
	have the largest eccentric connectivity index among all connected graphs
	of given order $n$. These results are summarized in Table \ref{table1}, where 
	\begin{itemize}
		\vspace{-0.2cm}\item $\K{n}$ is the complete graph of order $n$;
		\vspace{-0.2cm}\item $\Path{n}$ is the
		path of order $n$;
		\vspace{-0.2cm}\item $\ensuremath{{\sf W}_{n}}$ is the
		wheel of order $n$, i.e., the graph obtained by joining a vertex to all vertices of a cycle of order $n-1$;
		\vspace{-0.2cm}\item $\M{n}$ is the graph obtained from $\K{n}$ by removing
		a maximum matching and,
		if $n$ is odd, an additional edge adjacent to the unique vertex that still degree $n-1$;
		\vspace{-0.2cm}\item $\extG{n}{D}$ is
		the graph constructed from a path
		$u_0-u_1-\ldots-u_D$ by joining each vertex of a clique $\K{n-D-1}$
		to $u_0$, $u_1$ and $u_2$.
	\end{itemize}
	
\begin{table}[h!]
	\begin{center}
		\caption{Largest eccentric connectivity index for a fixed order $n$}
		\label{table1}
		
		\begin{tabular}{ccc}
				$n$  & optimal graphs \\
				\hline
				1&$\K{1}$\\
				2&$\K{2}$\\
				3&$\K{3}$ and $\Path{3}$\\
				4&$\M{4}$\\
				5&$\M{5}$ and $\ensuremath{{\sf W}_{5}}$\\
				6&$\M{6}$\\
				7&$\M{7}$\\
				8&$\M{8}$ and $\extG{8}{4}$\\
				$\geq 9$&$\extG{n}{\left\lceil\frac{n+1}{3}\right\rceil+1}$
						\end{tabular}
					\end{center}
				\end{table}

\noindent In addition to the above-mentioned graphs, we will also consider the following ones:
\begin{itemize}
	\vspace{-0.2cm}\item $\C{n}$ is the chordless cycle of order $n$;
	\vspace{-0.2cm}\item $\St{n}{x}$ is the graph of order $n$ obtained by linking all vertices of a stable set of $n-x$ vertices with all vertices of a clique $\K{x}$. The graph $\St{n}{1}$ is called a \emph{star}.
\end{itemize}
Also, for $n\geq 4$ and $p\leq n-3$, let  $\extH{n}{p}$ be the graph of order $n$ obtained by adding a dominating vertex (i.e., a vertex linked to all other vertices) to the graph or order $n-1$ having $p$ vertices of degree 0, and
		\begin{itemize}
			\vspace{-0.2cm}\item  $n-1-p$ vertices of degree 1 if $n-p$ is odd;
			\vspace{-0.2cm}\item  $n-2-p$ vertices of degree 1 and one vertex of degree 2 if $n-p$ is even.
		\end{itemize}

For illustration, $\extH{8}{3}$ and $\extH{9}{3}$ are drawn on Figure \ref{fig1}. Note that $\extH{4}{0}\simeq \St{4}{2}$. Moreover, $\extH{4}{0}$ has two dominating vertices while $\extH{4}{1}$ and $\extH{n}{p}$ have exactly one dominating vertex for all $n\geq 5$ and $p\leq n-3$.

\begin{figure}[h!]
	\centering\includegraphics[scale=0.75]{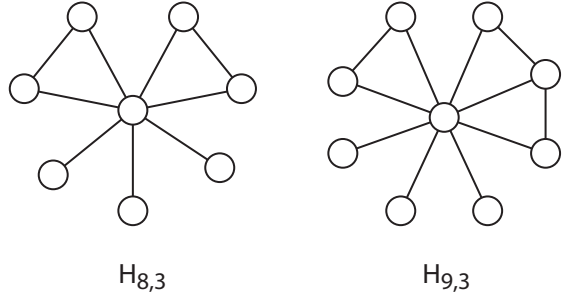}
	\captionof{figure}{Two graphs with $p=3$ pending vertices}.
	\label{fig1}
\end{figure}
In this paper, we first give an alternative proof to a result of Zhou and Du \cite{ZD10} showing that the stars are the only graphs with smallest eccentric connectivity index among all connected graphs of given order $n\geq 4$. These graphs have $n-1$ pending vertices (i.e., vertices of degree 1). We then consider all pairs $(n,p)$ of integers with $p\leq n-1$ and characterize the graphs with smallest eccentric connectivity index among all connected graphs of order $n$
with $p$ pending vertices.

\section{Minimizing $\eci$ for graphs with fixed order}

$\K{1}$ and $\K{2}$ are the only connected graphs with 1 and 2 vertices, respectively, while 
$\K{3}$ and $\Path{3}$ are the only connected graphs with 3 vertices. Since $\eci(\K{3})=\eci(\Path{3})=6$, all connected graphs of given order $n\leq 3$ have the same eccentric connectivity index. From now on, we therefore only consider connected graphs with fixed order $n\geq 4$. A proof of the following theorem was already given by Zhou and Du in \cite{ZD10}. Ours is slightly different.

\begin{thm}
	Let $G$ be a connected graph of order $n\geq 4$. Then $\eci(G) \ge 3(n-1)$, 
	with equality if and only if $G\simeq\St{n}{1}$.
\end{thm}
\begin{proof}
	Let $x$ be the number of dominating vertices (i.e., vertices of degree $n-1$) in $G$. We distinguish three cases.
	\begin{itemize}
		\vspace{-0.1cm}\item If $x=1$, then let $u$ be the dominating vertex in $G$. Clearly, $\ecc{u}=1$ and $\degr{u}=n-1$. All vertices $v\neq u$ have eccentricity $\ecc{v}=2$, while their degree is at least 1 (since $G$ is connected). Hence, $\eci(G)\geq (n-1)+2(n-1)=3(n-1)$, with equality if and only if all $v\neq u$ have degree 1, i.e., $G\simeq\St{n}{1}$.
		\vspace{-0.2cm}\item If $x>1$, then all dominating vertices $u$ have $\degr{u}\ecc{u}=n-1$, while all non-dominating vertices $v$ have $\degr{v}\geq x\geq 2$ and $\ecc{v}\geq 2$, which implies $\degr{u}\ecc{u}\geq 4$. If $n=4$, we therefore have $\eci(G)\geq 3n>3(n-1)$, while if $n>4$, we have $\eci(G)\geq 2(n-1)+4(n-2)=6n-10>3(n-1)$.
		\vspace{-0.2cm}\item If $x=0$, then every pending vertex $v$ has $\ecc{v}\geq 3$ since its only neighbor is a non-dominating vertex. Since the eccentricity of the non-pending vertices is at least two, we have $\degr{v}\ecc{v}\geq 3$ for all vertices $v$ in $G$, which implies $\eci(G)\geq 3n>3(n-1)$.
	\end{itemize}
\end{proof}

Stars have $n-1$ pending vertices. As will be shown in the next section, a similar result is more challenging when the total number of pending vertices is fixed to a value strictly smaller than $n-2$.

\section{Minimizing $\eci$ for graphs with fixed order and fixed number of pending vertices}

Let $G$ be a connected graph of order $n\geq 4$ with $p$ pending vertices.
Clearly, $p\leq n-1$, and $G\simeq\St{n}{1}$ if $p=n-1$. For $p=n-2$,  let $u$
and $v$ be the two non-pending vertices. Note that $u$ is adjacent to $v$ since
$G$ is connected. Clearly, $G$ is obtained by linking $x\leq n-3$ vertices of a stable set $S$ of $n-2$ vertices to $u$, and the $n-2-x$ other vertices of $S$ to $v$. The $n-2$ pending vertices $w$ have  $\degr{w}=1$ and $\ecc{w}=3$, while $\ecc{u}=\ecc{v}=2$ and $\degr{u}+\degr{v}=n$. Hence $\eci(G)=3(n-2)+2n=5n-6$ for all graphs of order $n$ with $n-2$ pending vertices. 

The above observations show that all graphs of order $n$ with a fixed number $p\geq n-2$ of pending vertices have the same eccentric connectivity index. As will be shown, this is not the case when $n\geq 4$ and $p\leq n-3$. We will prove that $\extH{n}{p}$ is almost always the unique graph minimizing the eccentric connectivity index. Note that

\[
\eci(\extH{n}{p})=\left\{\begin{array}{ll}
n-1+2p+4(n-p-1)=5n-2p-5 &\mbox{if }n-p\mbox{ is odd}  \\
n-1+2p+4(n-p-2)+6=5n-2p-3 &\mbox{if }n-p\mbox{ is even}.
\end{array}\right.
\]

%

\vspace{0.3cm}\begin{thm}\label{lem2}
	Let $G$ be a connected graph of order $n\geq 4$ with $p\leq n-3$ pending vertices and one dominating vertex. Then $\eci(G)\geq \eci(\extH{n}{p})$, with equality if and only if $G\simeq\extH{n}{p}$.
\end{thm}
\begin{proof}
	The dominating vertex $u$ in $G$ has $\degr{u}\ecc{u}=n-1$, the pending vertices $v$ have $\degr{v}\ecc{v}=2$, and the other vertices $w$ have $\ecc{w}=2$ and $\degr{w}\geq 2$. Hence, $\eci(G)$ is minimized if all non-pending and non-dominating vertices have degree 2, except one that has degree 3 if $n-p-1$ is odd. In other words, $\eci(G)$ is minimized if and only if $G\simeq\extH{n}{p}$. 
\end{proof}

\vspace{0.3cm}\begin{thm}\label{lem3}
	Let $G$ be a connected graph of order $n\geq 4$, with at least two dominating vertices. 
	\begin{itemize}
		\vspace{-0.2cm}\item If $n=4$ then $\eci(G)\geq 12$, with equality if and only if $G\simeq\K{4}$.
		\vspace{-0.2cm}\item If $n=5$ then $\eci(G)\geq 20$, with equality if and only if $G\simeq\St{5}{2}$ or $G\simeq\K{5}$.
		\vspace{-0.2cm}\item If $n\geq 6$ then $\eci(G)\geq 6n-10$, with equality if and only if $G\simeq\St{n}{2}$.
	\end{itemize}
\end{thm}
\begin{proof} 
	Let $x$ be the number of dominating vertices in $G$. Then $\degr{u}\ecc{u}=n-1$ for all dominating vertices $u$, while $\ecc{v}=2$ and $\degr{v}\geq x$ for all other vertices $v$. Hence, $\eci(G)\geq -2 x^2 + x(3 n-1)$.
	\begin{itemize}
		\vspace{-0.1cm}\item If $n=4$ then $\eci(G)\geq f(x)=-2 x^2 + 11x$. Since $2\leq x\leq 4$, $f(2)=14, f(3)=15$, and $f(4)=12$, we conclude that  $\eci(G)\geq 12$, with equality if and only if $x=4$, which is the case when $G\simeq\K{4}$.
		\vspace{-0.3cm}\item If $n=5$ then $\eci(G)\geq f(x)=-2 x^2 + 14x$. Since $2\leq x\leq 5$, $f(2)=f(5)=20$ and $f(3)=f(4)=24$, we conclude that  $\eci(G)\geq 20$, with equality if and only if $x=2$ or $5$, which is the case when $G\simeq\St{5}{2}$ or $G\simeq\K{5}$.
		\vspace{-0.3cm}\item If $n\geq 6$ then $-2 x^2 + x(3 n-1)$ is minimized for $x=2$, which is the case when $G\simeq\St{n}{2}$.
	\end{itemize}		
\end{proof}

\begin{thm}\label{lem4}
	Let $G$ be a connected graph of order $n\geq 4$, with $p\leq n-3$ pending vertices and no dominating vertex. Then $\eci(G)>\eci(\extH{n}{p})$ unless $n=5$, $p=0$ and $G\simeq\C{5}$, in which case $\eci(G)=\eci(\extH{n}{0})=20$.
\end{thm}
\begin{proof}
	Let $U$ be the subset of vertices $u$ in $G$ such that  $\degr{u}=\ecc{u}=2$. If $U$ is empty, then all non-pending vertices $v$ in $G$ have $\degr{v}\geq 2$ and $\ecc{v}\geq 2$ (since $G$ has no dominating vertex), and at least one of these two inequalities is strict, which implies $\degr{u}\ecc{u}\geq 6$. Also, every pending vertex $w$ has $\ecc{w}\geq 3$ since their only neighbor is not dominant. Hence, $\eci(G)\geq 6(n-p)+3p=6n-3p$. Since $p\leq n-3$, we have $\eci(G)\geq 5n-2p+3>\eci(\extH{n}{p})$.
	
	So, assume $U\neq \emptyset$. Let $u$ be a vertex in $U$, and let $v,w$ be its two neighbors. Also, let $A=N(v)\setminus (N(w)\cup\{w\})$, $B=(N(v)\cup N(w))\setminus \{u\}$, and $C=N(w)\setminus (N(v)\cup\{v\})$. Since $\ecc{u}=2$, all vertices of $G$ belong to $A\cup B\cup C\cup \{u,v,w\}$. We finally define $B'$ as the subset of $B$ that contains all vertices $b$ of $B$ with $\degr{b}=2$ (i.e., their only neighbors are $v$ and $w$).
	
	\vspace{0.5cm}\noindent\emph{Case 1}: $v$ is adjacent to $w$.
	
	\noindent $A\neq \emptyset$ else $w$ is a dominating vertex, and $C\neq \emptyset$ else $v$ is dominating. Let $G'$ be the graph obtained from $G$ by replacing every edge linking $v$ to a vertex $a\in A$ with an edge linking $w$ to $a$, and by removing all edges linking $v$ to a vertex of $B\setminus B'$. Clearly, $G'$ is also a connected graph of order $n$ with $p$  pending vertices, and $w$ is the only dominating vertex in $G'$. It follows from Theorem \ref{lem2} that $\eci(G')\geq \eci(\extH{n}{p})$. Also,
	\begin{itemize}
		\vspace{-0.1cm}\item $\degr{u}=\degrp{u}$ and $\ecc{u}=\eccp{u}$;
		\vspace{-0.2cm}\item $\degr{x}=\degrp{x}$ and $\ecc{x}\geq \eccp{x}$ for all $x\in A\cup C$;
		\vspace{-0.2cm}\item $\degr{x}=\degrp{x}$ and $\ecc{x}=\eccp{x}$ for all $x\in B'$;
		\vspace{-0.2cm}\item $\degr{x}>\degrp{x}$ and $\ecc{x}=\eccp{x}$ for all $x\in B\setminus B'$.		
	\end{itemize}
	Hence, \[
	\sum\limits_{x\in A\cup B\cup C\cup\{u\}}\degr{x}\ecc{x} \geq \sum\limits_{x\in A\cup B\cup C\cup\{u\}}
	\degrp{x}\eccp{x}.
	\]
	Moreover, 
		\begin{itemize}
			\vspace{-0.1cm}\item $\degr{v}\ecc{v} + \degr{w}\ecc{w} = 2(\sizs{A}+\sizs{B}+2) +
			2(\sizs{C}+\sizs{B}+2) = 2\sizs{A}+4\sizs{B}+2\sizs{C}+8$;
		\vspace{-0.2cm}\item 	$\degrp{v}\eccp{v}+\degrp{w}\eccp{w}=2(\sizs{B'}+2)+\sizs{A}+\sizs{B}+\sizs{C}+2.$
			\end{itemize}
			We therefore have 
			\[
			\begin{array}{rll}
		\eci(G)-\eci(G')&=&\quad\sum\limits_{x\in A\cup B\cup C\cup\{u\}}\degr{x}\ecc{x}+(\degr{v}\ecc{v} + \degr{w}\ecc{w})\\
		 &  & - \sum\limits_{x\in A\cup B\cup C\cup\{u\}}\degrp{x}\eccp{x} - (\degrp{v}\eccp{v}+\degrp{w}\eccp{w})\\
		 &\geq& (2\sizs{A}+4\sizs{B}+2\sizs{C}+8)-(2(\sizs{B'}+2)+\sizs{A}+\sizs{B}+\sizs{C}+2)\\
		 &=&\sizs{A}+\sizs{C}+3(\sizs{B'}+\sizs{B\setminus B'})-2\sizs{B'}+2\\
		 &=&\sizs{A}+\sizs{C}+\sizs{B'}+3\sizs{B\setminus B'}+2 > 0\\
			\end{array}			
			\]
	This implies $\eci(G)>\eci(G')\geq \eci(\extH{n}{p})$.
	
\vspace{0.5cm}\noindent\emph{Case 2}: $v$ is not adjacent to $w$, and both $A\cup (B\setminus B')$ and $C\cup (B\setminus B')$ are nonempty.
		
\noindent Let $G'$ be the graph obtained from $G$ by adding an edge linking $v$ to $w$, by replacing every edge linking $v$ to a vertex $a\in A$ with an edge linking $w$ to $a$, and by removing all edges linking $v$ to a vertex of $B\setminus B'$. Clearly, $G'$ is also a connected graph of order $n$ with $p$ pending vertices. As in the previous case, we have 
	\[
	\sum\limits_{x\in A\cup B\cup C\cup\{u\}}\degr{x}\ecc{x} \geq \sum\limits_{x\in A\cup B\cup C\cup\{u\}}
	\degrp{x}\eccp{x}.
	\]\\
	
\noindent Moreover, $\ecc{v}\geq 2$ and $\ecc{w}\geq 2$, while $\eccp{v}\leq 2$ and $\eccp{w}=1$, which implies
	\begin{itemize}
		\vspace{-0.1cm}\item $\degr{v}\ecc{v} + \degr{w}\ecc{w} \geq 2(\sizs{A}+\sizs{B}+1) +
		2(\sizs{C}+\sizs{B}+1)=2\sizs{A}+4\sizs{B}+2\sizs{C}+4$;
		\vspace{-0.1cm}\item 	$\degrp{v}\eccp{v}+\degrp{w}\eccp{w}\leq 2(\sizs{B'}+2)+\sizs{A}+\sizs{B}+\sizs{C}+2.$
	\end{itemize}
	We therefore have 
	\[
	\begin{array}{rll}
	\eci(G)-\eci(G')&\geq& (2\sizs{A}+4\sizs{B}+2\sizs{C}+4)-(2(\sizs{B'}+2)+\sizs{A}+\sizs{B}+\sizs{C}+2)\\
	&=&\sizs{A}+\sizs{C}+\sizs{B'}+3\sizs{B\setminus B'}-2.\\
	\end{array}			
	\]
	If $B\setminus B'\neq \emptyset$, $w$ is the only dominating vertex in $G'$, and $\eci(G)-\eci(G')>0$. It then follows from Theorem \ref{lem2} that $\eci(G)>\eci(G')\geq \eci(\extH{n}{p})$. So assume $B\setminus B'= \emptyset$. Since $A\cup (B\setminus B')\neq \emptyset$, and $C\cup (B\setminus B')\neq \emptyset$, we have $A\neq \emptyset$ and $C\neq \emptyset$. Hence, once again, $w$ is the only dominating vertex in $G'$, and we know from Theorem \ref{lem2} that $\eci(G')\geq \eci(\extH{n}{p})$. 
	\begin{itemize}
		\item If $\sizs{B'}\geq 1$, $\sizs{A}\geq 2$ or $\sizs{C}\geq 2$, then $\eci(G)>\eci(G')\geq \eci(\extH{n}{p})$.
		\item If $\sizs{B'}=0$ and $\sizs{A}=\sizs{C}=1$, there are two possible cases:
		\begin{itemize}
			\item if the vertex in $A$ is not adjacent to the vertex in $C$, then $n=5$, $p=2$, $G\simeq \Path{5}$ and $G'\simeq\extH{5}{2}$. Hence, $\eci(G)=24>16=\eci(\extH{n}{p})$;
			\item if the vertex in $A$ is adjacent to the vertex in $C$, then $n=5$, $p=0$, $G\simeq \C{5}$ and $G'\simeq\extH{5}{2}$. Hence, $\eci(G)=\eci(\extH{n}{p})=20$;
		\end{itemize}
		
	\end{itemize}

\vspace{0.4cm}\noindent\emph{Case 3}: $v$ is not adjacent to $w$, and 	at least one of $A\cup (B\setminus B')$ and $C\cup (B\setminus B')$ is empty.

\noindent Without loss of generality, suppose $A\cup (B\setminus B')=\emptyset$. We distinguish two subcases.

\vspace{0.4cm}\noindent\emph{Case 3.1}: $B'=\emptyset$.

\noindent Since $n\geq 4$, $C\neq \emptyset$. Also, since $p\leq n-3$, there is a non-pending vertex $r\in C$. Let $G'$ be the graph obtained from $G$ by removing the edge linking $u$ and $v$ and by linking $v$ to $w$ and to $r$. Note that $G'$ is a connected graph of order $n$ with $p$ pending vertices~: while $v$ was pending in $G$, but not $u$, the situation is the opposite in $G'$. Note also that Theorem \ref{lem2} implies $\eci(G')\geq \eci(\extH{n}{p})$ since $w$ is the only dominating vertex in $G'$. We then have:
	\begin{itemize}
		\vspace{-0.1cm}\item $\degr{u}\!=\!2$, $\degrp{u}\!=\!1$ and $\ecc{u}\!=\!\eccp{u}\!=\!2$, which gives $\degr{u}\ecc{u}-\degrp{u}\eccp{u}=2$;
		\vspace{-0.1cm}\item $\degr{v}\!=\!1$, $\degrp{v}\!=\!2$ $\ecc{v}\!=\!3$ and $\eccp{v}\!=\!2$, which gives $\degr{v}\ecc{v}-\degrp{v}\eccp{v}=\!-\!1$;
		\vspace{-0.1cm}\item $\degr{w}\!=\!n-2$, $\degrp{w}\!=\!n-1$ $\ecc{w}\!=\!2$ and $\eccp{w}\!=\!1$, which gives $\degr{w}\ecc{w}-\degrp{w}\eccp{w}=n-3$;
		\vspace{-0.1cm}\item $\degrp{r}\!=\!\degr{r}\!+\!1$, $\ecc{r}\!=\!3$ and $\eccp{w}\!=\!2$, which gives $\degr{r}\ecc{r}-\degrp{r}\eccp{r}=\degr{r}-2$;
		\vspace{-0.1cm}\item $\degrp{c}\!=\!\degr{c}$ and $\ecc{c}>\eccp{c}$ for all $c\in (C\setminus\{r\})$. Since $r$ has a neighbor in $C$ of degree at least 2, we have $\sum_{c\in C\setminus\{r\}}(\degr{c}\ecc{c}-\degrp{c}\eccp{c}\geq 2)$. 	\end{itemize}
Hence, $\eci(G)-\eci(G') \geq 2 - 1 + \underbrace{n-3}_{> 0} +
\underbrace{\degr{r}-2}_{\geq 0}+2 > 0$, which implies $\eci(G)>\eci(G')\geq \eci(\extH{n}{p})$.\\

\vspace{0.4cm}\noindent\emph{Case 3.2}: $B'\neq \emptyset$.

\noindent Let $b_1,\ldots,b_{\sizs{B'}}$ be the vertices in $B'$. Remember that the unique neighbors of these vertices are $v$ and $w$. Let $G'$ be the graph obtained from $G$ as follows. We first add an edge linking $v$ to $w$. Then, for every odd $i<\sizs{B'}$, we add an edge linking $b_i$ to $b_{i+1}$ and remove the edges linking $v$ to $b_i$ and to $b_{i+1}$.
We then have
\begin{itemize}
	\vspace{-0.2cm}\item $\degr{x}=\degrp{x}$ and $\ecc{x}=\eccp{x}$ for all $x\in B'\cup C\cup\{u\}$;
	\vspace{-0.2cm}\item $\degr{v}=\sizs{B'}+1$, $\degrp{v}\leq 3$,  $\ecc{v}\geq 2$, and $\eccp{v}\leq 2$;
	\vspace{-0.2cm}\item $\degr{w}=\sizs{B'}+\sizs{C}+1$, $\degrp{w}=\sizs{B'}+\sizs{C}+2$,  $\ecc{w}=2$, and $\eccp{w}=1$.
\end{itemize}
Hence, 
	\[
	\begin{array}{rll}
	\eci(G)-\eci(G')&=&\degr{v}\ecc{v}+\degr{w}\ecc{w}-\degrp{v}\eccp{v}+\degrp{w}\eccp{w}\\
	&\geq& 2(\sizs{B'}+1) + 2(\sizs{B'}+\sizs{C}+1) - 6-(\sizs{B'}+\sizs{C}+2)\\
	&=& 3\sizs{B'}+\sizs{C}-4.
	\end{array}			
\]
IF $\sizs{B'}\geq 2$ or $\sizs{C}\geq 2$, then $\eci(G)-\eci(G')>0$, and since $w$ is then the only dominating vertex in $G'$, we know from Theorem \ref{lem2} that $\eci(G)>\eci(G')\geq \eci(\extH{n}{p})$. So, assume  $\sizs{B'}= 1$ and $\sizs{C}\leq 1$:
\begin{itemize}
	\vspace{-0.2cm}\item if $\sizs{C}= 0$ then $n=4$, $p=0$, $G\simeq\C{4}$ and $G'\simeq\extH{4}{0}$ which implies $\eci(G)=16>14=\eci(\extH{n}{p})$;
	\vspace{-0.2cm}\item if $\sizs{C}= 1$ then $n=5$, $p=1$, $\eci(G)=23$ and $G'\simeq\extH{5}{1}$ which implies $\eci(G)>20=\eci(\extH{n}{p})$.
\end{itemize}
\end{proof}

We can now combine these results as follows. Assume $G$ is a connected graph of order $n$ with $p$ pending vertices. If $p\geq 1$, then $G$ has at most one dominating vertex, and it follows from Theorems \ref{lem2} and \ref{lem4} that $\extH{n}{p}$ is the only graph with maximum eccentric connectivity index. If $p=0$ and $n=4$, then $G$ cannot contain exactly one dominating vertex, and Theorems \ref{lem3} and \ref{lem4} show that $\K{4}$ is the only graph with maximum eccentric connectivity index. If $p=0$ and $n=5$, Theorems \ref{lem2}, \ref{lem3} and \ref{lem4} show that $\extH{5}{0}$, $\St{5}{2}$, $\K{5}$ and $\C{5}$ are the only candidates to minimize the eccentric connectivity index, and since $\eci(\extH{5}{0})=\eci(\St{5}{2})=\eci(\K{5})=\eci(\C{5})=20$, the four graphs are the optimal ones. If $p=0$ and $n\geq 6$ then we know from Theorems \ref{lem2}, \ref{lem3} and \ref{lem4} that $\St{n}{2}$ and $\extH{n}{0}$ are the only candidates to minimize the eccentric connectivity index. Since $\eci(\St{6}{2})=26<27=\eci(\extH{6}{0})$, $\eci(\St{7}{2})=32>30=\eci(\extH{7}{0})$ and $\eci(\St{n}{2})=6n-10>5n-3\geq\eci(\extH{n}{0})$ for $n\geq 8$, we deduce that $\St{6}{2}$ is the only graph with maximum eccentric connectivity index when $n=6$ and $p=0$, while $\extH{n}{0}$ is the only optimal graph when $n\geq 7$ and $p=0$. This is summarized in the following Corollary.

\begin{cor}
Let $G$ be a connected graph of order $n\geq 4$ with $p\leq n-3$ pending vertices. 
\begin{itemize}
	\vspace{-0.2cm}\item If $p\geq 1$ then $\eci(G)\geq \eci(\extH{n}{p})$, with equality if and only if $G\simeq \extH{n}{p}$;
	\vspace{-0.0cm}\item If $p=0$ then
	\vspace{-0.0cm}\begin{itemize}
		\item if $n=4$ then $\eci(G)\geq 12$, with equality if and only if $G\simeq \K{4}$;
	\vspace{-0.0cm}	\item if $n=5$ then $\eci(G)\geq 20$, with equality if and only if $G\simeq \extH{5}{0}$, $\St{5}{2}$, $\K{5}$ or $\C{5}$;
	\vspace{-0.0cm}	\item if $n=6$ then $\eci(G)\geq 26$, with equality if and only if $G\simeq \St{6}{2}$;
	\vspace{-0.0cm}	\item if $n\geq 7$ then $\eci(G)\geq \eci(\extH{n}{0})$, with equality if and only if $G\simeq \extH{n}{0}$.		
	\end{itemize} 
\end{itemize}  
\end{cor}
 
\section{Conclusion}
We have characterized the graphs with smallest eccentric connectivity index among those of fixed order $n$ and fixed or non-fixed number of pending vertices. Such a characterization for graphs with a fixed order $n$ and a fixed size $m$ was given in \cite{ZD10}. It reads as follows.
\begin{thm}Let $G$ be a connected graph of order $n$ with $m$ edges, where $n-1\leq m <{ {n}\choose{2}}.$ Also, let
	\[k=\left\lfloor\frac{2n-1-\sqrt{(2n-1)^2-8m}}{2}\right\rfloor.
	\]
	Then $\eci(G)\geq 4m-k(n-1)$, with equality if and only if $G$ has $k$ dominating vertices and $n-k$ vertices of eccentricity 2.
\end{thm}
It is, however, an open question to characterize the graphs with largest eccentric connectivity index among those of fixed  order $n$ and fixed size $m$. The following conjecture appears in \cite{Hau18}, where $\extg{n}{D}{k}$ is
the graph of order $n$ constructed from a path
$u_0-u_1-\ldots-u_D$ by joining each vertex of a clique $\K{n-D-1}$
to $u_0$ and $u_1$, and $k$ vertices of the clique to $u_2$.

\begin{conj}
	Let $G$ be a connected graph of order $n$ with $m$ edges, where $n-1\leq m \leq{ {n-1}\choose{2}}.$ Also, let
	\[D=\left\lfloor\frac{2n+1-\sqrt{17+8(m-n)}}{2}\right\rfloor \text{ and } k=m-{{n-D+1}\choose{2}}-D+1.
	\]
	Then $\eci(G)\leq \eci(\extg{n}{D}{k})$, with equality if and only if $G\simeq \extg{n}{D}{k}$ or $D=3$, $k=n-4$ and $G$ is the graph constructed from a path $u_0-u_1-u_2-u_3$,  by joining $1\leq i\leq n-3$ vertices of a clique $\K{n-4}$ to $u_0,u_1,u_2$ and the $n-4-i$ other vertices of $\K{n-4}$ to $u_1,u_2,u_3$.
\end{conj}

\bibliographystyle{acm}

\end{document}